\documentclass{llncs}

\usepackage{amsmath}
\usepackage{amssymb}
\usepackage{algorithm}
\usepackage{algorithmic}
\usepackage{multirow}
\usepackage[colorlinks,linkcolor=black,anchorcolor=black,citecolor=black]{hyperref}
\usepackage{comment}

\newcommand{\CONTINUE}{\textbf{continue}}
\newcommand{\BREAK}{\textbf{break}}
\newcommand{\bm}[1]{\boldsymbol{#1}}
\newcommand{\code}[1]{\texttt{#1}}

\newtheorem{thm}{Theorem}
\newtheorem{prop}[thm]{Proposition}
\newtheorem{eg}{Example}

\newcommand{\mmid}{\mathrm{mid}}
\newcommand{\rad}{\mathrm{rad}}
\newcommand{\midrad}{\mathrm{midrad}}

\begin{document}

\title{Real Root Isolation of Polynomial Equations Based on Hybrid Computation}
\author{Fei Shen\inst{1}\and Wenyuan Wu\inst{2} \and Bican Xia\inst{1}}
\institute{LMAM \& School of Mathematical Sciences, Peking University
\and Chongqing Institute of Green and Intelligent Technology, Chinese Academy of Sciences\\
 \email{shenfei@pku.edu.cn ~ wuwenyuan@cigit.ac.cn ~ xbc@math.pku.edu.cn} }
\date{}
\maketitle

\begin{abstract}
A new algorithm for real root isolation of polynomial equations
based on hybrid computation is presented in this paper.
Firstly, the approximate (complex) zeros of the given polynomial equations
are obtained via homotopy continuation method. Then, for each approximate zero,
an initial box relying on the Kantorovich theorem is constructed,
which contains the corresponding accurate zero.
Finally, the Krawczyk interval iteration with interval arithmetic
is applied to the initial boxes so as to check whether or not
the corresponding approximate zeros are real and to obtain the
real root isolation boxes.
Meanwhile, an empirical construction of initial
box is provided for higher performance. Our experiments on many benchmarks show
that the new hybrid method is more efficient, compared with the traditional symbolic approaches.
\end{abstract}

\section{Introduction}

The \emph{Real Roots Isolation} of polynomial equations is a procedure that
uses disjoint regions to isolate all the distinct real roots of polynomial equations,
with only one root in each region.
Formally speaking, let $\bm{F} = (f_1,f_2,\dots ,f_n)^{T}$ be polynomial equations defined on $\bm{R}^{n}$,
i.e. $f_i\in \bm{R}[x_1,x_2,\dots,x_n]$.
Suppose
    $\bm{F}(\bm{x}) = 0$
has only finite many real roots, say $\xi^{(1)},\xi^{(2)},\dots,\xi^{(m)}$.
The target of real root isolation is to compute a family of regions $S_1,S_2,\dots,S_m,\ S_j\subset\bm{R}^{n} (1\le j\le m)$,
such that $\xi^{(j)}\in S_j$ and $S_i\cap S_j = \emptyset\ (1\le i,j\le m)$.
Usually, we use rectangular boxes to denote the regions above.
So we often call these isolated boxes \emph{intervals} in this paper. Theoretically,
the width of intervals for some special problems can be very small.
Hence, we assume that the accuracy of numerical computation in this
paper can be arbitrarily high. However, it is also important to
point out that such case rarely happens and double-precision is
usually enough to obtain very small intervals in practice.

Real root isolation is an important problem in symbolic computation.
It can be viewed as a kind of exact algorithm for solving equations
since no root formula is available in general situation.
It is also a critical part of some other important algorithms,
such as CAD and real root classification for semi-algebraic systems, etc.
Improvement on real root isolation will benefit all of these algorithms.

We impose some hypothesis on the problem discussed here.
First is that the system is square,
i.e. the number of equations is the same as that of variables.
Then we only handle the systems with finite many roots.
Positive dimensional solution is beyond the scope of this paper.
Moreover, we suppose that the Jacobian matrix of $\bm{F}$ is non-singular
at each root of $\bm{F}(\bm{x}) = 0$.
So we only deal with the simple root cases.
For the singular situation, the deflation method \cite{Leykin2006111,deflation,wuzhiISS08,DaytonLiZeng2011}
can be applied, which is one of our ongoing work.

Most of the previous real root isolation algorithms are based on
symbolic computations. For instance, the Uspensky
algorithm \cite{RealZerosOfPoly} based on Descartes' rule is for
polynomials in one variable. In multi-variable scenario, we have
``First algorithm'' based on monotonicity \cite{xia2002} and ``Second algorithm''
based on ``upper-lower bound'' polynomial \cite{zhang2006}. There are
also some other algorithms based on different techniques, see for example
\cite{gao07,BCLM09,springerlink:10.1007/s002000050114,Cheng2012843}.

An advantage of those symbolic methods is that exact results can be obtained since they use symbolic
computation and some of them can be extended to semi-algebraic
systems. However, there are also some disadvantages. Some of these
method could only handle the isolation of complex zeros. And some of
them need to triangularize the system first, which is unacceptable
in computation when the system is complicated sometimes, such as
more variables or high degrees. While some methods that do not use
triangularization have to give a huge initial interval to include
all the real roots \cite{zhangting,zhangxiaoxia}, which is extremely
inefficient.

In order to avoid these problems and design a new algorithm that could efficiently solve more complicated systems
and provide accurate interval results, we employ hybrid computation to take both the advantages of symbolic and numerical methods.

The basic idea of this paper is to use numerical method
to obtain all the approximate zeros of polynomial systems, including possible non-real ones.
With these approximations, small initial intervals which contains the corresponding real roots are constructed.
We then apply symbolic method to these initial intervals to verify whether there is a real root in it or not.
The main method we use in numerical computation is homotopy continuation, and for symbolic process we use the Krawczyk iteration.

Most of the work in this paper comes from \cite{shenfei}.
In Section 2, we will introduce some preliminaries, including
homotopy continuation and interval arithmetic. A new real root
isolation algorithm is discussed in Section 3. To test our new
method, our experimental results on benchmarks together with comparison and analysis
will be presented in Section 4. Finally, there is a summary in Section 5 and some future work will also be discussed.

\section{Preliminary}

We introduce in this section some basic theories and tools
that would be used in our algorithm.

\subsection{Homotopy Continuation Method}

Homotopy continuation method is an important numerical computation method,
which is used in various fields.
We only treat it as an ``algorithm black box'' here,
where the input is a polynomial system, and the output is its approximate zeros.
Please find the details about the theory in \cite{homLi,homotopy}.

For our purpose, it is convenient to utilize some existing software,
such as Hom4ps-2.0 \cite{hom4ps}, PHCpack \cite{PHCpack} and HomLab \cite{homlab}.

In our implementation, we use Hom4ps-2.0, which could return all the approximate complex zeros
of a given polynomial system efficiently, along with residues and condition numbers.

\subsection{Interval arithmetic}

Interval arithmetic plays an important role
in real root isolation algorithms \cite{zhangting,zhangxiaoxia,zhang2006}.
The two main differences between our new algorithm and the traditional ones in \cite{zhangting,zhangxiaoxia} are:
1) Verification only carry out on the localized ``small'' intervals; 2) symbolic computation
is replaced with floating point numerical computation.

Most of the interval operations in this paper's algorithms are based on Rump's floating
point verification work \cite{Rump} and accomplished
by using the Matlab package Intlab \cite{intlab},
including interval arithmetic operations and Jacobian matrix, Hessian matrix calculations.\footnote[1]{
    See reference \cite{Rump}, Section 11, \emph{Automatic differentiation}.}

\subsubsection{Basic concepts}
We introduce some basic interval arithmetic theories in this section.
See reference \cite{interval} for more details.

For given numbers $\underline{x},\overline{x}\in\bm{R}$,
if $\underline{x}\le\overline{x}$, we call
\begin{displaymath}
    X = [\underline{x},\overline{x}] = \{x\in\bm{R}|\underline{x}\le x \le\overline{x}\}
\end{displaymath}
a \emph{bounded closed interval}, or \emph{interval} for short.
Denote by $I(\bm{R})$ the set of all the bounded close intervals on $\bm{R}$,
and $I(A) = \{ X\in I(\bm{R}) | X\subseteq A\}$ all the intervals on $A\subseteq\bm{R}$.
Especially, if $\underline{x} = \overline{x}$, we call $X$ a \emph{point interval}.

For intervals, there are some common quantities:
\begin{description}
    \item[midpoint] $\mmid(X) = (\underline{x} + \overline{x})/2$
    \item[width] $\mathrm{W}(X) = \overline{x} - \underline{x}$
    \item[radius] $\rad(X) = \frac{1}{2}\mathrm{W}(X)$
    \item[low end point] $\inf(X) = \underline{x}$
    \item[high end point] $\sup(X) = \overline{x}$
\end{description}
Obviously we have $X = [\mmid(X)-\rad(X),\mmid(X)+\rad(X)]$.
An interval is usually expressed by its midpoint and radius.
For example, if $m = \mmid(X)$, $r = \rad(X)$, then we can write the formula above as
$X = \midrad(m,r)$.

We can also define the arithmetic operations over intervals.
Let $X = [\underline{x},\overline{x}],Y = [\underline{y},\overline{y}]\in I(\bm{R})$,
\begin{itemize}
    \item[-] $X + Y = [\underline{x} + \underline{y}, \overline{x} + \overline{y}]$
    \item[-] $X - Y = [\underline{x} - \overline{y}, \overline{x} - \underline{y}]$
    \item[-] $X\cdot Y =
        [\min(\underline{x}\underline{y},\underline{x}\overline{y}, \overline{x}\underline{y}, \overline{x}\overline{y}),
        \max(\underline{x}\underline{y}, \underline{x}\overline{y}, \overline{x}\underline{y}, \overline{x}\overline{y})]$
    \item[-] $X / Y = [\underline{x},\overline{x}]\cdot [1/\overline{y},1/\underline{y}],\ 0\not\in Y$
\end{itemize}

A vector is called an \emph{interval vector} if all its components are intervals.
\emph{Interval matrix} can be similarly defined.
For interval vectors and interval matrices, the concepts such as midpoint, width, radius, etc,
and the arithmetic operations are defined in components.

Let $f:\bm{R}^{n}\rightarrow \bm{R}$ be a function,
if there exists an interval map
\begin{displaymath}
    F:I(\bm{R}^{n})\rightarrow I(\bm{R})
\end{displaymath}
such that for all $x_i\in X_i (i = 1,2,\dots,n)$,
\begin{displaymath}
    F([x_1,x_1],[x_2,x_2],\dots,[x_n,x_n]) = f(x_1,x_2,\dots,x_n)
\end{displaymath}
holds, then we call $F$ an \emph{interval expand} of $f$.

We call $F:I(\bm{R}^{n})\rightarrow I(\bm{R})$ an interval map with \emph{inclusive monotonicity}
if $\bm{X}\subseteq\bm{Y}$ implies $F(\bm{X}) \subseteq F(\bm{Y})$
for any given intervals $\bm{X}$ and $\bm{Y}$.
The definitions above can all be extended to the situations in $I(\bm{R}^{n})\rightarrow I(\bm{R}^{n})$.
And it is easy to prove that all the polynomial operations satisfy the inclusive monotonicity.

\subsubsection{Krawczyk operator}\label{section:K-intro}

The Krawczyk operator plays a key role in the real root verification of interval arithmetic.
The main accomplishment comes from the work of Krawczyk and Moore.
We only list some important results here. Complete proofs can be found in \cite{interval}.

Suppose $\bm{f}:D\subseteq\bm{R}^{n}\rightarrow\bm{R}^{n}$ is continuous differentiable on $D$.
Consider the equation
\begin{eqnarray}\label{eq:eqn}
    \bm{f}(\bm{x}) = 0.
\end{eqnarray}
Let $\bm{f}'$ be the Jacobi matrix of $\bm{f}$,
$\bm{F}$ and $\bm{F}'$ be the interval expand of $\bm{f}$ and $\bm{f}'$ with inclusive monotonicity, respectively.
For $\bm{X}\in I(D)$ and any $\bm{y}\in\bm{X}$, define the \emph{Krawczyk operator} as:
\begin{eqnarray}\label{eq:K1}
    K(\bm{y},\bm{X}) = \bm{y} - Y\bm{f}(\bm{y}) + (I - Y\bm{F}'(\bm{X}))(\bm{X} - \bm{y})
\end{eqnarray}
where $Y$ is any $n\times n$ non-singular matrix.

Especially, we assign $\bm{y} = \mmid(\bm{X})$, so Formula (\ref{eq:K1}) becomes
\begin{eqnarray}\label{eq:K2}
    K(\bm{X}) = \mmid(\bm{X}) - Y\bm{f}(\mmid(\bm{X})) + (I - Y\bm{F}'(\bm{X}))\rad(\bm{X})[-1,1].
\end{eqnarray}
Formula (\ref{eq:K2}) is often used in practice.

The reason why the Krawczyk operator is so important is that it has some nice properties.
\begin{prop}\label{thm:Kprop}
Suppose $K(\bm{y},\bm{X})$ is defined as Formula (\ref{eq:K1}), then
\begin{enumerate}
    \item If $\bm{x}^{*}\in\bm{X}$ is a root of Equation (\ref{eq:eqn}),
        then for any $\bm{y}\in\bm{X}$, we have $\bm{x}^{*}\in K(\bm{y},\bm{X})$;
    \item For any $\bm{y}\in\bm{X}$, if $\bm{X}\cap K(\bm{y},\bm{X}) = \emptyset$ holds,
        then there is no roots in $\bm{X}$;
    \item For any $\bm{y}\in\bm{X}$ and any non-singular matrix $Y$, if $K(\bm{y},\bm{X})\subseteq\bm{X}$ holds,
        then Equation (\ref{eq:eqn}) has a solution in $\bm{X}$;
    \item Moreover, for any $\bm{y}\in\bm{X}$ and any non-singular matrix $Y$,
        if $K(\bm{y},\bm{X})$ is strict inclusive in $\bm{X}$, then Equation (\ref{eq:eqn}) has only one root in $\bm{X}$.
\end{enumerate}
\end{prop}

With the properties above, we can easily develop a real root verification method,
which is a little different from the classical one, and will be explained later in this paper.

Meanwhile, with the hypothesis we set in introduction,
all the systems considered here are non-singular ones with only simple roots.
So the Jacobian matrix at the zeros are all invertible.
Thus, we often set $Y = (\mmid\bm{F}'(\bm{X}))^{-1}$ and the Krawczyk operator becomes
\begin{eqnarray}\label{eq:K-Moore}
    K(\bm{X}) & = & \mmid(\bm{X}) - (\mmid \bm{F}'(\bm{X}))^{-1}\bm{f}(\mmid(\bm{X}))\nonumber\\
     & & + (I - (\mmid \bm{F}'(\bm{X}))^{-1}\bm{F}'(\bm{X}))\rad(\bm{X})[-1,1].
\end{eqnarray}
This is also called the Moore form of the Krawczyk operator.

\section{Real root isolation algorithm}

In this section, we will present our new algorithm for real root isolation
based on hybrid computation.
As mentioned before, our idea is to construct the initial intervals
corresponding to the approximate zeros obtained by homotopy continuation,
then verify them via the Krawczyk interval iteration to obtain the isolation results.
In the end, we combine these sub-procedures to give the final algorithm description.

\subsection{Construction of initial intervals}

To apply the Krawczyk interval iteration,
obviously the construction of initial intervals
is a key procedure.
We should guarantee both the \emph{correctness} and \emph{efficiency},
that is, make sure the initial box contains the corresponding accurate real root,
and keep the interval radius as small as possible so as to shorten the iteration time.

Thus a valid error estimation for the initial approximate zeros
should be established.
And we discuss this issue in both theory and practice aspects here.

\subsubsection{Error estimation theory}

The core problem of the construction of initial box
is the choice of interval radius,
which is indeed an error estimation for the approximate zero.
There are dozens of error analysis for this question,
from classic results to modern ones, especially about the Newton method.
For example, in \cite{ComplexityAndRealComputation},
S. Smale et al. gave a detailed analysis.
However, their method requires computation of high order derivatives,
which is not so convenient for our problem.

Here we employ the Kantorovich Theorem to give our error estimation.

\begin{thm}[Kantorovich]\label{thm:Kan}
Let $X$ and $Y$ be Banach spaces and $F:D\subseteq X\rightarrow Y$ be an operator,
which is Fr\'echet differentiable on an open convex set $D_0\subseteq D$.
For equation $F(x) = 0$, if the given approximate zero $x_0\in D_0$ meets the following three conditions:
\begin{enumerate}
    \item $F'(x_0)^{-1}$ exists, and there are real numbers $B$ and $\eta$ such that
        \begin{displaymath}
        \| F'(x_0)^{-1} \| \le B, \quad \| F'(x_0)^{-1}F(x_0) \| \le \eta,
        \end{displaymath}
    \item $F'$ satisfies the Lipschitz condition on $D_0$:
        \begin{displaymath}
        \| F'(x)-F'(y) \| \le K\| x - y \|, \ \forall x,y\in D_0,
        \end{displaymath}
    \item $h = BK\eta \le \frac{1}{2}, \  O(x_0,\frac{1 - \sqrt{1-2h}}{h}\eta)\subset D_0,$
\end{enumerate}
then we claim that:
\begin{enumerate}
  \item $F(x) = 0$ has a root $x^{*}$ in $\overline{O(x_0,\frac{1 - \sqrt{1-2h}}{h}\eta)}\subset \overline{D_0}$,
        and the sequence $\{x_k:x_{k+1} = x_k - F'(x_k)^{-1}F(x_k)\}$ of Newton method
        converges to $x^{*}$;
  \item For the convergence of $x^{*}$, we have:
         \begin{eqnarray}
            \| x^{*} - x_{k+1} \| \le
            \frac{\theta^{2^{k+1}}(1-\theta^{2})}{\theta(1-\theta^{2^{k+1}})}\eta
        \end{eqnarray}
        where $\theta = \frac{1 - \sqrt{1-2h}}{1 + \sqrt{1-2h}}$;
  \item The root $x^{*}$ is unique in $\overline{D_0}\cap \overline{O(x_0,\frac{1 + \sqrt{1-2h}}{h}\eta)}$.
\end{enumerate}
\end{thm}
In the theorem, $O(x,r)$ denotes the ball neighborhood  whose center is $x$ and radius is $r$,
and $\overline{O(x,r)}$ refers to the closure of the ball neighborhood.
The proof can be found in \cite{KanThm}.

Since the approximation $x_0$ is already the result of homotopy process,
what we care about is the initial interval w.r.t. $x_0$, i.e. the proper upper bound for $\| x^{*} - x_0 \|$.
So we have the following proposition, which is a direct corollary of the Kantorovich Theorem.

\begin{prop}\label{thm:rad}
Let $\bm{F} = (f_1,f_2,\dots,f_n)^{T}$ be a polynomial system, where $f_i \in \bm{R}[x_1,x_2,\dots,x_n]$.
Denote by $\bm{J}$ the Jacobian matrix of $\bm{F}$. For an approximation $\bm{x}_0\in\bm{C}^n$,
if the following conditions hold:
\begin{enumerate}
\item $\bm{J}^{-1}(\bm{x}_0)$ exists, and there are real numbers $B$ and $\eta$ such that
    \begin{displaymath}
    \| \bm{J}^{-1}(\bm{x}_0) \| \le B, \quad \| \bm{J}^{-1}(\bm{x}_0)\bm{F}(\bm{x}_0) \| \le \eta,
    \end{displaymath}
\item There exists a ball neighbourhood $O(\bm{x}_0,\omega)$ such that $\bm{J}(\bm{x})$ satisfies the Lipschitz condition on it:
    \begin{displaymath}
    \|\bm{J}(\bm{x})-\bm{J}(\bm{y})\| \le K\|\bm{x} - \bm{y}\|,\
        \forall\bm{x},\bm{y}\in O(\bm{x}_0,\omega)
    \end{displaymath}
\item Let $h = BK\eta$,
    \begin{displaymath}
     h  \leq \frac{1}{2}, \mbox{ and } \omega \ge \frac{1 - \sqrt{1-2h}}{h}\eta,
    \end{displaymath}
\end{enumerate}
then $\bm{F}(\bm{x}) = 0$ has only one root $\bm{x}^{*}$ in
$\overline{O(\bm{x}_0,\omega)}\cap \overline{O(\bm{x}_0,\frac{1 + \sqrt{1-2h}}{h}\eta)}$.
\end{prop}

\begin{proof}
We consider $F$ as an operator on $\bm{C}^n\rightarrow\bm{C}^n$,
obviously it is Fr\'echet differentiable, and from
\begin{displaymath}
    \bm{F}(\bm{x}+\bm{h}) = \bm{F}(\bm{x}) + \bm{J}(\bm{x})\bm{h} + o(\bm{h})
\end{displaymath}
we can get
\begin{displaymath}
    \lim_{\bm{h}\rightarrow 0}
    { \frac{\|\bm{F}(\bm{x}+\bm{h}) - \bm{F}(\bm{x}) - \bm{J}(\bm{x})\bm{h}\|}{\|\bm{h}\|} } = 0.
\end{displaymath}
Thus the first order Fr\'echet derivative of $\bm{F}$ is just the Jacobian matrix $\bm{J}$, i.e. $\bm{F}'(\bm{x}) = \bm{J}(\bm{x})$.
So by Theorem \ref{thm:Kan}, the proof is completed immediately after checking the situation of $\| \bm{x^{*}} - \bm{x_0} \|$.
\end{proof}

It is easy to know that
$\frac{1 - \sqrt{1-2h}}{h} \le 2$.
So we can just assign $\omega = 2\eta$. Then we need to check whether
$BK\eta \leq \frac{1}{2}$ in the neighborhood $O(\bm{x}_0,2\eta)$. Even
though the initial $\bm{x}_0$ does not satisfy the conditions, we
can still find a proper $\bm{x}_k$ after several Newton iterations,
since $B$ and $K$ are bounded and $\eta$ will approach zero.
And we only need to find an upper bound for the Lipschitz constant $K$.

\subsubsection{Constructive algorithm}

Now we will give a constructive procedure for the Lipschitz constant
$K$.

Let $J_{ij} = \partial{f_i}/\partial{x_j}$,
apply mean value inequality\cite{mean-value-ineqn} to each element of $\bm{J}$ on $O(\bm{x}_0,\omega)$ to get
    \begin{eqnarray}\label{eq:mean-value}
    \|J_{ij}(\bm{y}) - J_{ij}(\bm{x})\| \le
    \sup_{\kappa_{ij}\in line(\bm{x},\bm{y})}{\| \nabla{J_{ij}(\kappa_{ij})} \|}
    \cdot \|\bm{y} - \bm{x}\|,\ \forall \bm{x},\bm{y}\in O(\bm{x}_0,\omega)
    \end{eqnarray}
where $\nabla = (\partial/\partial{x_1},\partial/\partial{x_2},\dots,\partial/\partial{x_n})$ is the gradient operator
and $line(\bm{x},\bm{y})$ refers to the line connecting $\bm{x}$ with $\bm{y}$.
Since $\nabla{J}$ is continous, we can find a $\zeta_{ij}\in line(\bm{x}, \bm{y})$ such that
$\|\nabla{J_{ij}}(\zeta_{ij})\| = \sup_{\kappa_{ij}\in line(\bm{x},\bm{y})}{\| \nabla{J_{ij}(\kappa_{ij})} \|}$.
So we get
    \begin{eqnarray}
    \|J_{ij}(\bm{y}) - J_{ij}(\bm{x})\| \le
    \| \nabla{J_{ij}(\zeta_{ij})} \|
    \cdot \|\bm{y} - \bm{x}\|,\ \forall \bm{x},\bm{y}\in O(\bm{x}_0,\omega)
    \end{eqnarray}
Setting $\big(\|\nabla{J_{ij}(\zeta_{ij})}\|\cdot\|\bm{y} - \bm{x}\|\big)_{n\times n} = \triangle{\bm{J}}$, then
$\|J(y) - J(x)\| \le \|\triangle{\bm{J}}\|$.
And for $\triangle{\bm{J}}$ we have
    \begin{eqnarray} \label{eq:est1}
    \|\triangle{\bm{J}}\|_{\infty}
    & = & \| \big( \|\nabla{J_{ij}(\zeta_{ij})}\|_{\infty}
    \|\bm{y} - \bm{x}\|_{\infty} \big)_{n\times n} \|_{\infty} \nonumber\\
    & \le & \| \big( \|\nabla{J_{ij}(\zeta_{ij})}\|_{\infty} \big)_{n\times n} \|_{\infty}
    \cdot \|\bm{y} - \bm{x}\|_{\infty}\nonumber\\
    & = & \max_{1\le i\le n}{\sum_{j=1}^{n}{\|\nabla{J_{ij}(\zeta_{ij})}\|_{\infty}}}
    \cdot \|\bm{y} - \bm{x}\|_{\infty}
    \end{eqnarray}
Note that $\nabla{J_{ij}(\zeta_{ij})}$ is a vector, so if we use
$|\cdot|_{\max}$ to denote the maximum module component of a vector,
then we have
    \begin{eqnarray}\label{eq:le1}
    \|\triangle{\bm{J}}\|_{\infty} & \le &
    \max_{1\le i\le n}{\sum_{j=1}^{n}{|\nabla{J_{ij}(\zeta_{ij})}|_{\max}}}
    \cdot\|\bm{y}-\bm{x}\|_{\infty}.
    \end{eqnarray}

Let $H_i =
(\frac{\partial^{2}{f_i}}{\partial{x_j}\partial{x_k}})_{n\times n}$
be the Hessian matrix of $f_i$, and let $H_i =
(h_{1}^{(i)},\dots,h_{n}^{(i)})$, where $h_{j}^{(i)}$
are the column vectors. Then we have
    \begin{eqnarray}\label{eq:dJ-estimate}
    \| \triangle{\bm{J}} \|_{\infty}
    \le \max_{1\le i\le n}{ \sum_{j=1}^{n}{ |h_{j}^{(i)}(\zeta_{ij})|_{\max} } }
    \cdot \|\bm{y} - \bm{x}\|_{\infty}.
    \end{eqnarray}

For convenience, we construct $\bm{X}_0 = \midrad(\bm{x}_0,\omega)$
with $\bm{x}_0$ as centre and $\omega = 2\eta$ as radius.

Now we have $|h_{j}^{(i)}(\zeta_{ij})|_{\max} \le |h_{j}^{(i)}(\bm{X}_0)|_{\max}$. So
    \begin{eqnarray}
    \| \triangle{\bm{J}} \|_{\infty}
    & \le & \max_{1\le i\le n}{ \sum_{j=1}^{n}{ |h_{j}^{(i)}(\bm{X}_0)|_{\max} } }
    \cdot \|\bm{y} - \bm{x}\|_{\infty}.
    \end{eqnarray}
Therefore
    \begin{eqnarray}\label{eq:LipConst}
    K = \max_{1\le i\le n}{ \sum_{j=1}^{n}{ |h_{j}^{(i)}(\bm{X}_0)|_{\max} } }
    \end{eqnarray}
is the Lipschitz constant  w.r.t. $\bm{J}$.

Now we give an algorithm for computing initial intervals in
Algorithm \ref{alg:initwidth}.

\begin{algorithm}[htb]
\caption{init\_width}
\label{alg:initwidth}
\begin{algorithmic}[1]
\REQUIRE~~ Equation $F$; Approximation $x_0$; Number of variables $n$
\ENSURE~~ Initial interval's radius $r$

\REPEAT
    \STATE $\bm{x}_0 = \bm{x}_0 - \bm{J}^{-1}(\bm{x}_0)\bm{F}(\bm{x}_0)$;
    \STATE $B = \|\bm{J}^{-1}(\bm{x}_0)\|_{\infty}$;\
        $\eta = \|\bm{J}^{-1}(\bm{x}_0)\bm{F}(\bm{x}_0)\|_{\infty}$;
    \STATE $\omega = 2\eta$;
    \STATE $\bm{X}_0 = \midrad(\bm{x}_0,\omega)$;
    \STATE $K = 0$;
    \FOR{$i = 1\ \TO\ n $}
        \STATE Compute the Hessian matrix
            $H_i = (h_{1}^{(i)},h_{2}^{(i)},\dots,h_{n}^{(i)})$ of $\bm{F}$ on $\bm{X}_0$;
        \IF {$\sum_{j=1}^{n}{ |h_{j}^{(i)}(\bm{X}_0)|_{\max} } > K$}
            \STATE $K = \sum_{j=1}^{n}{ |h_{j}^{(i)}(\bm{X}_0)|_{\max} }$;
        \ENDIF
    \ENDFOR
    \STATE $h = BK\eta$;
\UNTIL{$h\le 1/2$}
\RETURN $r = \frac{1-\sqrt{1-2h}}{h}\eta$
\end{algorithmic}
\end{algorithm}

\subsection{Empirical estimation}

As so far, we have established a rigorous method to construct
initial intervals. This method takes a complex approximate zero
as input to obtain an initial box. But in practice we often find
many approximations with ``large'' imaginary parts which strongly
indicate that they are non-real. A natural question is
\begin{quote}
    Can we detect these non-real roots without using interval arithmetic?
\end{quote}

Let $\bm{z}$ be an approximation of the real root $\bm{\xi}$. Because
\begin{displaymath}
\| \mathfrak{Re(\bm{z})} - \bm{\xi} \| \le \| \bm{z} - \bm{\xi} \|,
\end{displaymath}
then we can see the real part $\mathfrak{Re(\bm{z})}$ is also an approximation of this root
and is even closer.
So we can simply replace $\bm{x_0}$ by $\mathfrak{Re(\bm{x_0})}$ in Algorithm \ref{alg:initwidth} to construct the initial box.

The other consideration is the efficiency of numerical computation.
When we use Proposition \ref{thm:rad}, lots of interval matrix
operations would be executed, which cost much time than the point
operations. So if we can find an empirical estimate radius, which can be
computed much faster, but is still valid for most of the equations,
then that will be a good choice in practice.

We now give one such empirical estimation.

For $\bm{F} = 0$, let $\bm{x}^{*}$ be an accurate root and $\bm{x}_0$ be its approximation.
Although the mean value theorem is not valid in complex space, the Taylor expansion is still valid.
And the polynomial systems considered here are all continuous,
so we suppose the equation satisfies the mean value theorem approximately:
\begin{eqnarray}
    0 = \bm{F}(\bm{x}^{*}) \approx \bm{F}(\bm{x}_0) + \bm{J}(\bm{\xi})(\bm{x}^{*} - \bm{x}_0)
\end{eqnarray}
where $\bm{\xi}$ is between $\bm{x}^{*}$ and $\bm{x}_0$. So we have
\begin{eqnarray*}
    \bm{x}^{*} - \bm{x}_0 \approx - \bm{J}^{-1}(\bm{\xi})\bm{F}(\bm{x}_0).
\end{eqnarray*}
Let $\bm{J}(\bm{\xi}) = \bm{J}(\bm{x}_0) + \triangle{\bm{J}}$, then
\begin{eqnarray}\label{eq:est2}
    \bm{J}(\bm{\xi}) = \bm{J}(\bm{x}_0)(I + \bm{J}^{-1}(\bm{x}_0)\triangle{\bm{J}}),\nonumber\\
    \bm{J}^{-1}(\bm{\xi}) = (I + \bm{J}^{-1}(\bm{x}_0)\triangle{\bm{J}})^{-1}\bm{J}^{-1}(\bm{x}_0).
\end{eqnarray}
For $\triangle{\bm{J}}$, we can get an estimation similar to Formula (\ref{eq:dJ-estimate}):
\begin{eqnarray}
    \| \triangle{\bm{J}} \|_{\infty}
    \le \max_{1\le i\le n}{ \sum_{j=1}^{n}{ |h_{j}^{(i)}(\zeta_{ij})|_{\max} } }
    \cdot \|\bm{x}^{*} - \bm{x}_0\|_{\infty}.
\end{eqnarray}
From our hypothesis, $\bm{x}^{*}$ and $\bm{x}_0$ are very close, so are $\zeta_{ij}$ and $\bm{x}_0$.
Thus, we approximate $\bm{x}_0$ with $\zeta_{ij}$.
Meanwhile, from $\bm{x}_0$, after a Newton iteration,
we get $\bm{x}_1 = \bm{x}_0 - \bm{J}^{-1}(\bm{x}_0)\bm{F}(\bm{x}_0)$.
Thus we may consider that the distance between $\bm{x}^{*}$ and $\bm{x}_0$
is more or less the same with that of $\bm{x}_0$ and $\bm{x}_1$,
so we replace $\|\bm{x}^{*} - \bm{x}_0\|$
with $\| \bm{x}_1 - \bm{x}_0 \|= \| \bm{J}^{-1}(\bm{x}_0)\bm{F}(\bm{x}_0) \|$ for approximation.

So we get
\begin{eqnarray}
    \| \triangle{\bm{J}} \|_{\infty}
    \le \max_{1\le i\le n}{ \sum_{j=1}^{n}{ |h_{j}^{(i)}(\bm{x}_0)|_{\max} } }
    \cdot \| \bm{J}^{-1}(\bm{x}_0)\bm{F}(\bm{x}_0) \|_{\infty}.
\end{eqnarray}
Let $\lambda = \max_{1\le i\le n}{ \sum_{j=1}^{n}{ |h_{j}^{(i)}(\bm{x}_0)|_{\max} } }$, then
\begin{displaymath}
    \| \bm{J}^{-1}(\bm{x}_0)\triangle{\bm{J}} \|_{\infty} \le
        \lambda \| \bm{J}^{-1}(\bm{x}_0)\|_{\infty}^{2} \|\bm{F}(\bm{x}_0) \|_{\infty}.
\end{displaymath}
Because $\|\bm{F}(\bm{x}_0) \|_{\infty} \ll 1$, the last formula is also far less than 1.
So substitute that into Formula (\ref{eq:est2}) we can get
\begin{eqnarray}
    \|\bm{J}^{-1}(\bm{\xi})\|_{\infty}  \le
    \frac{\|\bm{J}^{-1}(\bm{x}_0)\|_{\infty}}
    {1 - \lambda \| \bm{J}^{-1}(\bm{x}_0)\|_{\infty}^{2} \|\bm{F}(\bm{x}_0) \|_{\infty}}.
\end{eqnarray}
Finally we obtain the empirical estimation
\begin{eqnarray}\label{eq:rad-estimate}
    \|\bm{x}^{*}-\bm{x}_0\|_{\infty} &\approx& \|\bm{J}^{-1}(\bm{\xi})\bm{F}(\bm{x}_0)\|_{\infty}\nonumber\\
    &\le& \frac{\|\bm{J}^{-1}(\bm{x}_0)\|_{\infty} \|\bm{F}(\bm{x}_0)\|_{\infty}}
    {1 - \lambda \| \bm{J}^{-1}(\bm{x}_0)\|_{\infty}^{2} \|\bm{F}(\bm{x}_0) \|_{\infty}}.
\end{eqnarray}

Notice that the inequality (\ref{eq:rad-estimate}) is only a non-rigorous estimation.
All the computation in it are carried out in a point-wise way, so it is faster than Proposition \ref{thm:rad}.
In the numerical experiments later we will see that this empirical estimate radius performs very well.
So we can use it to detect those non-real roots rather than the interval arithmetic.
We describe that in Algorithm \ref{alg:iscomplex}.

\begin{algorithm}[htb]
\caption{iscomplex}
\label{alg:iscomplex}
\begin{algorithmic}[1]
\REQUIRE~ Equation $F$; Approximation $\bm{z}$;
\ENSURE~ \code{true} ($\bm{z}$ must be non-real), or \code{false} ($\bm{z}$ may be real).

\STATE Compute Formula (\ref{eq:rad-estimate}), denote the result by $r'$;
\IF{ \code{any}( $\mathfrak{Im}(\bm{z}) > r'$ ) }
    \RETURN \code{true};
    \COMMENT{not a real root, continue to judge others}
\ELSE
    \RETURN \code{false};
    \COMMENT{may be a real root, call interval arithmetic to verify}
\ENDIF
\end{algorithmic}
\end{algorithm}
In Algorithm \ref{alg:iscomplex}, \code{any}() is a default function in Matlab,
which returns true if there is non-zero component in a vector.

\subsection{Krawczyk-Moore interval iteration}

We now discuss about the real root verification with a given interval.
In section \ref{section:K-intro}, we have introduced the Krawczyk operator.
With the properties in Proposition \ref{thm:Kprop},
we can determine whether an interval contains a real root
by the relationship of the original interval and the one after the Krawczyk iteration.

However, in practice, we can't expect the intervals to be entire inclusion or disjoint
after just one iteration.
Partly intersection is the most common cases that we encounter.
Since the real root is still in the interval after the Krawczyk iteration,
a normal method is to let $\bm{X}\cap K(\bm{X})$
be the new iteration interval.
So suppose $\bm{X}^{(0)}$ is the initial interval,
the iteration rule is $\bm{X}^{(k+1)} = \bm{X}^{(k)}\cap K(\bm{X}^{(k)})$,
where $K(\bm{X}^{(k)})$ is defined by Formula (\ref{eq:K-Moore}).
This update rule can make sure that the size of $\bm{X}^{(k)}$ is non-increasing.
But a problem is once we encounter $K(\bm{X}^{(k)})\cap\bm{X}^{(k)} == \bm{X}^{(k)}$,
the iteration will be trapped into endless loop.
So we have to divide $\bm{X}^{(k)}$ if this happened.

Thus, we introduce a bisection function \code{divide}().
To ensure the convergence of our algorithm, we divide the longest dimension of an interval vector.

This strategy may not be the optimal choice when the system's dimension is high.
Greedy method or optimization algorithm will be studied in future work.

We now give a formal description of \code{divide} function in Algorithm \ref{alg:divide}
and the Krawczyk-Moore iteration process in Algorithm \ref{alg:Krawczyk}.
\begin{algorithm}[htb]
\caption{divide}
\label{alg:divide}
\begin{algorithmic}[1]
\REQUIRE~ Interval vector $\bm{X}$
\ENSURE~ $\bm{X}^{(1)}$ and $
\bm{X}^{(2)}$, a decomposition of $\bm{X}$

\STATE Let $\bm{X}_i$ be the coordinate with the largest width in $\bm{X}$
\STATE $\bm{X}^{(1)} = \bm{X};\ \bm{X}^{(2)} = \bm{X}$;
\STATE $\bm{X}^{(1)}_i = [\inf(\bm{X}_i),\mmid(\bm{X}_i)]$;
\STATE $\bm{X}^{(2)}_i = [\mmid(\bm{X}_i),\sup(\bm{X}_i)]$;
\RETURN $\bm{X}^{(1)},\ \bm{X}^{(2)}$
\end{algorithmic}
\end{algorithm}

\begin{algorithm}[htb]
\caption{Krawczyk}
\label{alg:Krawczyk}
\begin{algorithmic}[1]
\REQUIRE~ $F$; initial box $\bm{X}$; isolation boxes $real\_roots$; number of real roots $nreal$
\ENSURE~ symbol of whether there is a real root $flag$; $real\_roots$; $nreal$

\STATE $Y = \mmid(\bm{F}'(\bm{X}))^{-1}$; $\bm{X}_t = K(\bm{X})$, where $K(\bm{X})$ is define by Formula (\ref{eq:K-Moore});
\IF{$\bm{X}_t\cap\bm{X} == \emptyset$}
    \RETURN $flag = \code{false}$;
\ENDIF
\WHILE{\NOT($\bm{X}_t\subseteq\bm{X}$)}
    \IF{$\bm{X}_t\cap\bm{X} == \bm{X}$}
        \STATE $[\bm{X}^{(1)},\bm{X}^{(2)}] = \code{divide}(\bm{X})$;
        \STATE $[f1,real\_roots,nreal] = \code{Krawczyk}(F,\bm{X}^{(1)},real\_roots,nreal)$;
        \IF{$f1 == \code{false}$}
            \STATE $[f2,real\_roots,nreal] = \code{Krawczyk}(F,\bm{X}^{(2)},real\_roots,nreal)$;
        \ENDIF
        \RETURN $f1\ \OR\ f2$;
    \ENDIF
    \STATE $\bm{X} = \bm{X}_t\cap\bm{X}$;
    \STATE $Y = ( \mmid\bm{F}'(\bm{X}))^{-1}$.
    \STATE $\bm{X}_t = K(\bm{X})$;
    \IF{$\bm{X}_t\cap\bm{X} == \emptyset$}
        \RETURN flag = \code{false};
    \ENDIF
\ENDWHILE
\STATE $nreal = nreal + 1$;
\STATE $real\_roots[nreal] = \bm{X}_t$
\RETURN $flag = \code{true}, real\_roots, nreal$;
\end{algorithmic}
\end{algorithm}

\subsection{Verification and refinement}

After the Krawczyk iteration, we already have all the real root isolated intervals,
but these are not the final results.
Since we require an isolation of disjoint intervals, we have to check the possible overlaps.

On the other hand, some intervals may not be as small as required by users, so we can
narrow them via bisection method until they match the requirement.

We discuss these details in this subsection.

\subsubsection{Remove the overlaps}
There is a basic hypothesis: for non-singular systems, each root has
an approximation, and from this approximation, the iteration will
end up in its corresponding accurate root, not any other zero. So we
only have to remove the overlaps, and the number of real roots won't
change.

However, we want to expand our algorithm into multi-roots cases.
And in that situation, it is possible that two isolated intervals contain the same real zero.
So whether or not the overlap part contains a real root, our algorithm has its corresponding processes.
See Algorithm \ref{alg:disjoint} for details.
\begin{algorithm}[htb]
\caption{disjoint\_process}
\label{alg:disjoint}
\begin{algorithmic}[1]
\REQUIRE~ Isolated intervals $real\_roots$; number of real roots $nreal$; $F$
\ENSURE~ Checked isolated intervals $real\_roots$; $nreal$

\STATE $k = 0$;
\FOR{$i = 1\ \TO\ nreal$}
    \STATE $\bm{X} = real\_roots[i]\ $; $new\_root = \code{true}$;
    \FOR{$j = 1\ \TO\ k$}
        \STATE $\bm{Y} = real\_roots[j]$;
        \STATE $\bm{Z} = \bm{X}\cap\bm{Y}$;
        \IF{$\bm{Z} == \emptyset$}
            \STATE \CONTINUE;
        \ENDIF
        \STATE $flag$ = \code{Krawczyk}($F$,$\bm{Z}$);
        \IF{$flag$ == \code{true}}
            \STATE $new\_root = \code{false}\ $; \BREAK;
        \ELSE
            \STATE $\bm{X} = \bm{X}\setminus\bm{Z}$;
            \STATE $real\_roots[j] = real\_roots[j]\setminus\bm{Z}$;
        \ENDIF
    \ENDFOR
    \IF{$new\_root == \code{true}$}
        \STATE $k = k + 1\ $; $real\_roots[k] = \bm{X}$
    \ENDIF
\ENDFOR
\RETURN $real\_roots\ $,$nreal = k$;
\end{algorithmic}
\end{algorithm}

The function \code{Krawczyk}() in Algorithm \ref{alg:disjoint} is a
little bit different from that in the Krawczyk-Moore iteration. In
the Krawczyk-Moore iteration, we have to store the information of
isolated real root intervals, so the $real\_roots$ and $nreal$ are
in the function arguments. However, we only need to know whether
there is a real root here, so only the symbol variable $flag$ is
returned. The situation is the same in Algorithm
\ref{alg:narrowing}.

\subsubsection{Narrow the width of intervals}
As we said in the introduction, the real root isolation can be
viewed as a kind of solving equations. And the width of the isolated
intervals is just like the accuracy of solutions. Similar to the
former algorithms, we can require the program to return an answer in
specified range. The difference is the symbolic algorithm can get
any precision they want in theory, but our floating point number
calculation can't beat the machine precision. In fact, in the Matlab
environment that we implement our algorithm, the resulted width
won't be smaller than the system zero threshold\footnote[1]{
    In Matlab2008b that we do the experiments, the zero threshold is 2.2204e-016.
}.
\begin{algorithm}[h]
\caption{narrowing}
\label{alg:narrowing}
\begin{algorithmic}[1]
\REQUIRE~ Isolated intervals $real\_roots$; Number of real roots $nreal$; $F$; Threshold $\tau$
\ENSURE~ $real\_roots$ after bisection

\FOR{$i = 1\ \TO\ nreal$}
    \STATE $\bm{X} = real\_roots[i]$;
    \WHILE{\code{any}($\rad(\bm{X}) > \tau$)}
        \STATE $[\bm{Y}^{(1)},\bm{Y}^{(2)}] = \code{divide}(\bm{X})$;
        \STATE $flag = \code{Krawczyk}(F,\bm{Y}^{(1)})$;
        \IF{$flag == \code{true}$}
            \STATE $\bm{X} = \bm{Y}^{(1)}$;
        \ELSE
            \STATE $\bm{X} = \bm{Y}^{(2)}$
        \ENDIF
    \ENDWHILE
    \STATE $real\_roots[i] = \bm{X}$;
\ENDFOR
\RETURN $real\_roots$
\end{algorithmic}
\end{algorithm}

We also use bisection to do the narrowing job.
Since there is only one root in the interval,
we only have to continue dividing and checking the half that contains that zero.
Formal description is in Algorithm \ref{alg:narrowing}.

\begin{algorithm}[htb]
\caption{real\_root\_isolate}
\label{alg:main}
\begin{algorithmic}[1]
\REQUIRE~ Equation $F(\bm{x)}$; number of variables $n$; Threshold $\tau$;
\ENSURE~ Isolated intervals of $F(\bm{x}) = 0$ and number of real roots $nreal$

\STATE [$complex\_roots,ncomplex$] = \code{hom4ps}($F$,$n$);\\
\STATE Initialize $real\_roots$ to be empty; $nreal = 0$;
\FOR{$i = 1$ $\TO$  $ncomplex$}
    \STATE $\bm{z}$ = $complex\_roots[i]$;
    \IF{$\code{iscomplex}(F,\bm{z})$}
        \STATE \CONTINUE;
    \ENDIF
    \STATE $r$ = \code{init\_width}[$F$,$\bm{z}$,$n$];
    \STATE $X_0 = \midrad(\mathfrak{Re}(\bm{z}),r)$;
    \STATE [$flag,real\_roots,nreal$] = \code{Krawczyk}($F,X_0,real\_roots,nreal$);
\ENDFOR
\STATE $[real\_roots,nreal] = \code{disjoint\_process}(real\_roots,nreal,F)$;
\STATE $real\_roots = \code{narrowing}(real\_roots,nreal,F,\tau)$;
\RETURN $real\_roots,\ nreal$;
\end{algorithmic}
\end{algorithm}

\subsection{Algorithm description}
Up to now, we have discussed all the parts of real root isolation algorithm in detail.
We give 
the final main program in Algorithm \ref{alg:main}.

\section{Experiments}

Now we apply our new method to some polynomial systems
and do some comparison with some former algorithms.

All the experiments are undertaken in Matlab2008b, with Intlab \cite{intlab} of Version 6.
For arbitrarily high accuracy, we can call Matlab's vpa (variable precision arithmetic), but in fact all the real roots of the examples below are isolated by using Matlab's
default double-precision floating point. We use Hom4ps-2.0 \cite{hom4ps} as our homotopy continuation tool to obtain initial approximate zeros.
Since computation time will be listed below, the computer information is also given here:
OS: Windows Vista, CPU: Inter\textregistered Core 2 Duo T6500 2.10GHz, Memory: 2G.

\subsection{Demo example}
We begin our illustration with a simple example.
\begin{eg}
Consider the real root isolation of the system below.
\begin{displaymath}
    \left\{ \begin{array}{rcl}
        x^{3} y^{2} + x + 3 & = & 0\\
        4 y z^{5} + 8 x^{2} y^{4} z^{4} - 1 & = & 0\\
        x + y + z - 1 & = & 0
        \end{array} \right.
\end{displaymath}
\end{eg}
The homotopy program tells us this system has 28 complex zeros in total.
And we get the following results after calling our \code{real\_root\_isolate} program.\\
{\small \code{
intval  =\\
{}[  -0.94561016957416,  -0.94561016957415]\\
{}[   1.55873837303161,   1.55873837303162]\\
{}[   0.38687179654254,   0.38687179654255]\\
intval  =\\
{}[  -1.18134319868123,  -1.18134319868122]\\
{}[  -1.05029487815439,  -1.05029487815438]\\
{}[   3.23163807683560,   3.23163807683561]\\
intval  =\\
{}[  -2.99999838968782,  -2.99999838968781]\\
{}[   0.00024421565895,   0.00024421565896]\\
{}[   3.99975417402886,   3.99975417402887]\\
intval  =\\
{}[  -0.79151164911096,  -0.79151164911095]\\
{}[   2.11038450699949,   2.11038450699950]\\
{}[  -0.31887285788855,  -0.31887285788854]\\
The order of variables:\\
    'x'\\
    'y'\\
    'z'\\
The number of real roots: 4
}
}

We verify the answers above with the DISCOVERER \cite{discoverer} package under Maple,
which also return $4$ isolated real roots. Here we show its output in floating point number format, i.e.\\{\small {}
[[-2.999998391, -2.999998389], [0.2442132427e-3, 0.2442180230e-3], [3.999754090, 3.999754249]], \\{}
[[-1.181343199, -1.181343199],[-1.050294975, -1.050294818],\allowbreak [3.231637836, 3.231638372]], \\{}
[[-.9456101805, -.9456101656], [1.558738033, 1.558738728],[.3868716359, .3868719935]], \\{}
[[-.7915116549, -.7915116400],[2.110384024, 2.110385000], [-.3188729882, -.3188727498]].\\{}
}
And we can see the answers perfectly match the ones of our program.

We list some information during the calculation of our algorithm here for reference.
Only the 4 real ones are given, and the other non-real ones
are all detected by our empirical estimation method.
We mention that all the imaginary parts of complex roots are significant larger
than the initial radius of our algorithm in order of magnitude in this example.

\begin{table}[htb]
\begin{tabular}{|c|c|c|c|c|}
\hline
        &   root1  &   root2    &   root3  &   root4\\
\hline
$B$      &   1.060227  &  1.192159   &  2.000864  &   0.874354\\
\hline
$K$     &   14.941946  &   7.198937e+003  & 4.095991e+003  &  16.988990 \\
\hline
$\eta$     &  4.260422e-016  &  4.20807e-016   &  8.882333e-016 &  5.764449e-016 \\
\hline
$h$    &  2.024791e-014   &  1.083446e-011    &  2.183861e-011 &  2.568823e-014 \\
\hline
estimate-rad        &  4.274976e-016   &  4.208067e-016  &  8.882344e-016   &  5.779921e-016 \\
\hline
empirical-rad        &  1.015249e-015   &  1.29164e-012  &  2.156138e-015 &  1.559270e-015 \\
\hline
\end{tabular}
\caption{Key quantities comparison}
\label{tab:demo}
\end{table}

We give some remarks on Table \ref{tab:demo}.
In the first row, \emph{root1} to \emph{root4} are refer to
the 4 real roots mentioned above respectively.
And $B,K,\eta,h$ are exactly the same as they are defined in algorithm \ref{alg:initwidth}.
The \emph{estimate-rad} are the radius obtained via algorithm \ref{alg:initwidth},
while the \emph{empirical-rad} are refer to the ones calculated by Formula (\ref{eq:rad-estimate}).

We say a little more words about the \emph{empirical-rad}.
Firstly, although the empirical ones are basically larger than the rigorous error radius,
they are still small enough, which hardly have any influence on the efficiency of interval iteration.
We will see this in the comparison experiments later.
But avoiding of interval matrix computation is very helpful to the algorithm.
Secondly, the radius obtained from Algorithm \ref{alg:initwidth} are so small that
they are even comparable to the zero threshold of Matlab system\footnote[1]{
  As mentioned before, the zero threshold in Matlab2008b is 2.2204e-016,
  which is almost the same order of magnitude of those radiuses.
}. And this could bring some uncertainty of floating point operation to our algorithm,
such as misjudgement of interval inclusion in Intlab, etc.
So we intend to use empirical estimation bound in next experiments.

For \emph{cyclic6} (see Appendix \ref{chp:cases}), the classic symbolic algorithm can do nothing
due to the difficulty of triangularization.
Meanwhile, we can easily get the 24 isolated real roots intervals
with our \code{real\_root\_isolate} program.

\subsection{Comparison experiment}

Many benchmarks have been checked with our \code{real\_root\_isolate} program.
Since the time complexity of both triangularization and homotopy continuation
are difficult to be analyzed, we mainly focus on the isolation results and the program execution time.

We investigate over 130 benchmarks provided by Hom4ps \cite{benchmark},
among which about 40 equations are non-singular systems.
We apply our program to these equations and all the experiments receive the right answers.
Here we list a few of them (see Appendix \ref{chp:cases} for details).

\begin{table}[htb]
\begin{tabular}{|c|c|c|c|c|}
\hline
Equation      &   total roots  &   real roots    &   DISCOVERER  &   complex roots detected\\
\hline
barry       &   20  &   2  &  2  &   18\\
\hline
cyclic5     &   70  &   10  & 10  &  60 \\
\hline
cyclic6     &  156  &  24   &  N/A &  132 \\
\hline
des18\_3    &  46   &  6    &  N/A &  40 \\
\hline
eco7        &  32   &  8  &  8   &  24 \\
\hline
eco8        &  64   &  8  &  N/A &  56 \\
\hline
geneig      &  10   &  10  & N/A  & 0  \\
\hline
kinema      &  40   &  8   & N/A  & 32  \\
\hline
reimer4     &  36   &  8   &  8  &  28 \\
\hline
reimer5     & 144   &  24  &  N/A & 120  \\
\hline
virasoro    & 256   &  224 &  N/A  & 32  \\
\hline
\end{tabular}
\caption{Real root isolation results comparison}
\label{tab:roots-ans}
\end{table}

The column \emph{real roots} in Table \ref{tab:roots-ans}
tells the number of intervals that our program isolated.
Compared with the results of DISCOVERER,
the new algorithm indeed works out all equations that are beyond
the capability of classic symbolic algorithm.
Moreover, the last column show that our empirical estimate method
detects all the non-real roots successfully.

\begin{table}[htb]
\begin{tabular}{|c|c|c|c|c|}
\hline
Equations      &   Total time  &   Homotopy time    &   Interval time  &   DISCOVERER\\
\hline
barry       &  0.421203 &  0.093601 &  0.327602 &  0.063 \\
\hline
cyclic5     &  2.948419 &  0.218401 &  2.652017 &   0.624\\
\hline
cyclic6     &  9.984064 & 0.639604  & 9.063658  & N/A  \\
\hline
des18\_3    &  4.180827 & 0.702004 &  3.385222 &  N/A \\
\hline
eco7        &  2.371215 &  0.265202 & 2.012413  &  15.881 \\
\hline
eco8        &  3.946825 &  0.499203 & 3.354022  &  N/A \\
\hline
geneig      &  4.243227 &  0.249602 & 3.868825  &  N/A \\
\hline
kinema      &  3.946825 &  1.014006 & 2.808018  &  N/A \\
\hline
reimer4     &  2.480416 &  0.374402 & 2.059213  &  24.711 \\
\hline
reimer5     &  12.963683 &  3.073220 & 9.578461  & N/A  \\
\hline
virasoro    &  137.124879 & 4.570829  & 109.996305  & N/A  \\
\hline
\end{tabular}
\caption{Execution time comparison, unit:s}
\label{tab:roots-time}
\end{table}

Table \ref{tab:roots-time} gives the comparison of program execution time.
The total time is not equal to the sum of homotopy time and interval iteration time
because we only count the CPU time, and there are other tasks such as I/O, format transform, etc.

Table \ref{tab:roots-time} also shows that interval iterations consume more time than homotopy continuation.
The reason is complicated and we enumerate some here:
\begin{enumerate}
    \item The homotopy continuation focuses only on floating-point number, while the Krawczyk iteration cares about intervals;
    \item Hom4ps-2.0 is a software complied from language C,
        which is much more efficient than the tool that we use to implement our algorithm, say Matlab.
    \item The interval iteration time increases as roots number grows since we examine the approximate zeros one by one.
        So the parallel computation of homotopy is much faster.
\end{enumerate}
We believe that with efficient language such as C/C++, and parallel computation,
the implementation of our algorithm will be much faster.

In order to verify our idea and see whether parallelization could help,
we go into every approximate zero's iteration process.
Some critical data are recorded in Table \ref{tab:roots-step}.
The \emph{avg. rad. of ans} is the average radius of the final isolated intervals,
while the \emph{avg. rad. of init.} indicates the average radius of the initial intervals.
The average time of each zero's interval iteration is shown in column \emph{avg. time of iteration}
along with the max interval iteration time in \emph{ max time of iter}.
We think the consumption for each zero's process is acceptable.

\begin{table}[htb]
\begin{tabular}{|c|c|c|c|c|}
\hline
Equation      &   avg.rad.of ans &   avg.rad.of init.    &   avg.time of iter   &   max time of iter\\
\hline
barry       &  3.552714e-015 &  1.377800e-014 &  0.054600 &  0.062400 \\
\hline
cyclic5     &  1.614703e-009 &  7.142857e-007 &  0.113881 &  0.140401 \\
\hline
cyclic6     &  4.440892e-016 &  2.137195e-015 &  0.183951 &  0.234002 \\
\hline
des18\_3    &  3.768247e-007 &  9.737288e-007 &  0.241802 &  0.296402 \\
\hline
eco7        &  1.998401e-015 &  1.483754e-013 &  0.122851 &  0.156001 \\
\hline
eco8        &  2.109424e-015 &  3.283379e-013 &  0.183301 &  0.218401 \\
\hline
geneig      &  2.664535e-016 &  5.721530e-014 &  0.315122 &  0.436803 \\
\hline
kinema      &  1.998401e-015 &  6.784427e-011 &  0.157951 &  0.218401 \\
\hline
reimer4     &  1.110223e-016 &  1.258465e-014 &  0.122851 &  0.156001 \\
\hline
reimer5     &  1.110223e-016 &  4.754080e-014 &  0.195001 &  0.421203 \\
\hline
virasoro    &  9.472120e-009 &  2.265625e-006 &  0.387844 &  0.624004 \\
\hline
\end{tabular}
\caption{Detail data for each iteration, unit:s}
\label{tab:roots-step}
\end{table}

From Table \ref{tab:roots-step} we can see that the initial interval radii are extremely small,
which leads to a nice process time for each iteration.
We point out that almost all real root checks are done by just \emph{one} Krawczyk iteration,
and hardly any overlap is found after all the Krawczyk iteration processes
due to the small initial intervals that we give.
All of these save a great deal of executing time of our program.

\section{Conclusion}

For the non-singular polynomial systems with variables' number
equals equations' number, this paper presents a new algorithm for
real root isolation based on hybrid computation. The algorithm first
applies homotopy continuation to obtain all the initial approximate
zeros of the system. For each approximate zero, an initial interval
which contains the corresponding accurate root is constructed. Then
the Krawczyk operator is called to verify all the initial intervals
so as to get all the real root isolated boxes. Some necessary check
and refinement work are done after that to ensure the boxes are
pairwise disjoint and meet width requirement.

In the construction of initial intervals, we give a rigorous radius
error bound based on the corollary of the Kantorovich theorem. Some
constructive algorithms are presented for both real and complex
approximate zeros. Meanwhile, we introduce an empirical estimate
radius, which has a nice performance in numerical experiments.

In the modification and implementation of the Krawczyk iteration
algorithm, some problems of interval arithmetic are also discussed
in this paper.

At last we utilize some existing tools to implement our algorithm
under Matlab environment. Many benchmarks have been checked along
with comparison and analysis.

We also mention some possible future work here.
The construction of initial intervals is still too complicated and
further optimization shall be studied. Also the empirical estimation
with more efficiency and accuracy is a considerable question. The
divide strategy in the Krawczyk iteration could also be improved,
which may be helpful in the high dimension cases.

In the aspect of implementation, replacing the Matlab implementation with C/C++ codes may improve the performance of our program.
Parallel computation can solve another bottleneck of our problem.
And for some small systems, or equations with special property,
the classic symbolic algorithm could be even faster.
So the tradeoff of symbolic and numerical computation is also an interesting direction.

\section*{Acknowledgements}
The work is partly supported by the ANR-NSFC project EXACTA (ANR-09-BLAN-0371-01/60911130369), NSFC-11001040, NSFC-11271034
and the project SYSKF1207 from ISCAS.
The authors especially thank professor Dongming Wang for the early
discussion on this topic in 2010 and also thank professor T. Y. Li
for his helpful suggestions and his team's work on Hom4ps2-Matlab
interface. Thank the referees for their valuable comments.


\begin{thebibliography}{99}

\bibitem{benchmark}
Benchmarks: http://hom4ps.math.msu.edu/HOM4PS\_soft\_files/equations.zip.


\bibitem{ComplexityAndRealComputation}
Blum L., Cucker F., Shub M. and Smale S.:
Complexity and Real Computation.
Springer, 1997.

\bibitem{BCLM09}
Boulier F., Chen C.B., Lemaire F. and Moreno Maza. M.:
 Real Root Isolation of Regular Chains.
In: Proceedings of the 2009 Asian Symposium on Computer Mathematics (ASCM 2009), 15--29, 2009.


\bibitem{Cheng2012843}
Cheng J.S., Gao X.S. and Guo L.L.:
Root isolation of zero-dimensional polynomial systems with linear univariate representation.
Journal of Symbolic Computation, 47(7):843--858, 2012.


\bibitem{gao07}
Cheng J.-S., Gao X.-S. and Yap C.-K.: Complete Numerical Isolation of
Real Zeros in Zero-dimensional Triangular Systems. In Proc.
ISSAC'2007, 92--99, 2007.


\bibitem{RealZerosOfPoly}
Collins G.E. and Loos R.:
Real zeros of polynomials.
In: Computer Algebra: Symbolic and Algebraic Computation,
(Buchberger B., Collins G.E. and Loos R. eds.), Springer-Verlag,
83--94, 1982.

\bibitem{DaytonLiZeng2011}
Dayton B., Li T.Y. and Zeng Z.G.:
 Multiple zeros of nonlinear systems, Math. Comp., 80:2143--2168, 2011.

\bibitem{KanThm}
Gragg G.W. and Tapia R.A.:
Optimal Error Bounds for the Newton-Kantorovich Theorem.
SIAM Journal on Numerical Analysis, 11(1):10--13, 1974.

\bibitem{Leykin2006111}
Leykin A., Verschelde J. and Zhao A.L.:
Newton's method with deflation for isolated singularities of polynomial systems.
Theoretical Computer Science, 359:111--122, 2006.

\bibitem{deflation}
Leykin A., Verschelde J. and Zhao A.L.:
Higher-order Deflation for Polynomial Systems With Isolated Singular Solutions.
In: Algorithms in Algeraic Geometry (Dickenstein, A. and Schreyer, F. and Sommese, A., eds.),
Springer, 2008, 79--97.

\bibitem{homLi}
Li T.Y.:
Numerical solution of multivariate polynomial systems by homotopy continuation methods.
Acta Numerica, 6:399--436, 1997.

\bibitem{hom4ps}
Li T.Y.: HOM4PS-2.0.
http://hom4ps.math.msu.edu/HOM4PS\_soft.htm, 2008.

\bibitem{interval}
Moore R.E., Kearfott R.B. and Cloud M.J.: Introduction to
Interval Analysis. Society for Industrial and Applied
Mathematics, Philadelphia, 2009.


\bibitem{mean-value-ineqn}
Mujica J.: Complex Analysis in Banach Spaces.
North-Holland Mathematics Studies, 120:99--138, Elsevier, 1986.

\bibitem{springerlink:10.1007/s002000050114}
Rouillier F.: Solving Zero-Dimensional Systems Through the Rational Univariate Representation.
Applicable Algebra in Engineering, Communication and Computing,
9:433--461, 1999.

\bibitem{intlab}
Rump {S.M.}: INTLAB - INTerval LABoratory.
In: Developments~in~Reliable Computing (Tibor Csendes ed.),
77--104, Kluwer Academic Publishers, 1999.
http://www.ti3.tu-harburg.de/rump/

\bibitem{Rump}
Rump S.M.: Verification methods: Rigorous results using floating-point arithmetic.
Acta Numerica, 19:287--449, 2010.


\bibitem{shenfei}
Shen F.: The Real Roots Isolation of Polynomial System Based on Hybrid Computation.
Master degree thesis, Peking University, April, 2012.


\bibitem{homotopy}
Sommese A. and Wampler C.:
The Numerical Solution of Systems of Polynomials: Arising in Engineering and Science.
World Scientific, 2005.

\bibitem{PHCpack}
Verschelde J.: PHCpack.\\
http://homepages.math.uic.edu/~jan/PHCpack/phcpack.html, 1999.

\bibitem{homlab}
Wampler C.: HomLab.
http://nd.edu/~cwample1/HomLab/main.html, 2005.


\bibitem{wuzhiISS08}
Wu X.L. and Zhi L.H.: Computing the multiplicity structure from geometric involutive form.
In: Proc. 2008 Internat. Symp. Symbolic Algebraic Comput. (ISSAC'08), 325--332, 2008.

\bibitem{discoverer}
Xia B.C.: DISCOVERER: A tool for solving semi-algebraic systems.
ACM Communications in Computer Algebra, 41(3):102--103, 2007.

\bibitem{zhang2006}
Xia B. and Zhang T.: Real Solution Isolation Using Interval Arithmetic.
Computers and Mathematics with Applications, 52:853--860, 2006.


\bibitem{xia2002}
Yang L. and Xia B.: An algorithm for isolating the real solutions of semi-algebraic systems.
Journal of Symbolic Computation, 34:461--477, 2002.






\bibitem{zhangting}
Zhang T.: Isolating Real Roots of Nonlinear Polynomial.
Master degree thesis, Peking University, 2004.

\bibitem{zhangxiaoxia}
Zhang T. and Xiao R. and Xia B.C.:
Real Solution Isolation Based on Interval Krawczyk Operator.
In: Proc. the 7th Asian Symposium on Computer Mathematics (ASCM 2005)
(Pae Sung-il and Park H. eds.), 235--237, 2005.





\end{thebibliography}



\appendix
\section{Benchmarks}\label{chp:cases}

\raggedbottom 
\begin{enumerate}
\item barry: Number of variables:3,Number of equations:3,Max degree:5
    \begin{eqnarray*}
     - x^{5} + y^{5} - 3 y - 1 & = & 0\\
    5 y^{4} - 3 & = & 0\\
     - 20 x + y - z & = & 0\\
    \end{eqnarray*}

\item cyclic5: Number of variables:5,Number of equations:5,Max degree:5
    \begin{eqnarray*}
    x_{1} + x_{2} + x_{3} + x_{4} + x_{5} & = & 0\\
    x_{1} x_{2} + x_{2} x_{3} + x_{3} x_{4} + x_{4} x_{5} + x_{1} x_{5} & = & 0\\
    x_{1} x_{2} x_{3} + x_{2} x_{3} x_{4} + x_{3} x_{4} x_{5} + x_{4} x_{5} x_{1} + x_{5} x_{1} x_{2} & = & 0\\
    x_{1} x_{2} x_{3} x_{4} + x_{2} x_{3} x_{4} x_{5} + x_{3} x_{4} x_{5} x_{1} + x_{4} x_{5} x_{1} x_{2} + x_{5} x_{1} x_{2} x_{3} & = & 0\\
    x_{1} x_{2} x_{3} x_{4} x_{5} - 1 & = & 0\\
    \end{eqnarray*}

\item cyclic6: Number of variables:6,Number of equations:6,Max degree:6
    \begin{eqnarray*}
    x_{0} + x_{1} + x_{2} + x_{3} + x_{4} + x_{5} & = & 0\\
    x_{0} x_{1} + x_{1} x_{2} + x_{2} x_{3} + x_{3} x_{4} + x_{4} x_{5} + x_{5} x_{0} & = & 0\\
    x_{0} x_{1} x_{2} + x_{1} x_{2} x_{3} + x_{2} x_{3} x_{4} + x_{3} x_{4} x_{5} + x_{4} x_{5} x_{0} + x_{5} x_{0} x_{1} & = & 0\\
    x_{0}x_{1}x_{2}x_{3} + x_{1}x_{2} x_{3} x_{4} + x_{2} x_{3} x_{4} x_{5} + x_{3} x_{4} x_{5}x_{0} & &\\
     + x_{4} x_{5} x_{0} x_{1} + x_{5} x_{0} x_{1} x_{2} & = & 0\\
    x_{0} x_{1} x_{2} x_{3} x_{4} + x_{1} x_{2} x_{3} x_{4} x_{5} + x_{2} x_{3} x_{4} x_{5} x_{0} + x_{3} x_{4} x_{5} x_{0} x_{1} & &\\
    + x_{4} x_{5} x_{0} x_{1} x_{2} + x_{5} x_{0} x_{1} x_{2} x_{3} & = & 0\\
    x_{0} x_{1} x_{2} x_{3} x_{4} x_{5} - 1 & = & 0\\
    \end{eqnarray*}

\item des18\_3: Number of variables:8,Number of equations:8,Max degree:3
    \begin{eqnarray*}
    15 a_{33} a_{10} a_{21} - 162 a_{10}^{2} a_{22} - 312 a_{10} a_{20} + 24 a_{10} a_{30} + 27 a_{31} a_{21} & &\\
     + 24 a_{32} a_{20} + 18 a_{22} a_{10} a_{32} + 30 a_{22} a_{30} + 84 a_{31} a_{10} & = & 0\\
    28 a_{22} a_{10} a_{33} + 192 a_{30} + 128 a_{32} a_{10} + 36 a_{31} a_{20} + 36 a_{33} a_{20} & &\\
     - 300 a_{10} a_{21} + 40 a_{32} a_{21} - 648 a_{10}^{2} + 44 a_{22} a_{31} & = & 0\\
    180 a_{33} a_{10} - 284 a_{22} a_{10} - 162 a_{10}^{2} + 60 a_{22} a_{32} + 50 a_{32} a_{10} & &\\
     + 70 a_{30} + 55 a_{33} a_{21} + 260 a_{31} - 112 a_{20} & = & 0\\
    6a_{33} a_{10} a_{20} + 10 a_{22} a_{10} a_{31} + 8 a_{32} a_{10} a_{21} - 162 a_{10}^{2} a_{21} & &\\
     + 16 a_{21} a_{30} + 14 a_{31} a_{20} + 48 a_{10} a_{30} & = & 0\\
    4 a_{22} a_{10} a_{30} + 2 a_{32} a_{10} a_{20} + 6 a_{20} a_{30} - 162 a_{10}^{2} a_{20} + 3 a_{31} a_{21} a_{10} & = & 0\\
    66 a_{33} a_{10} + 336 a_{32} + 90 a_{31} + 78 a_{22} a_{33} - 1056 a_{10} - 90 a_{21} & = & 0\\
     - 240 a_{10} + 420 a_{33} - 64 a_{22} + 112 a_{32} & = & 0\\
    136 a_{33} - 136 & = & 0\\
    \end{eqnarray*}

\item eco7: Number of variables:7,Number of equations:7,Max degree:3
    \begin{eqnarray*}
    x_{7} x_{1} + x_{7} x_{1} x_{2} + x_{7} x_{2} x_{3} + x_{7} x_{3} x_{4} + x_{7} x_{4} x_{5} + x_{7} x_{5} x_{6} - 1 & = & 0\\
    x_{7} x_{2} + x_{7} x_{1} x_{3} + x_{7} x_{2} x_{4} + x_{7} x_{3} x_{5} + x_{7} x_{6} x_{4} - 2 & = & 0\\
    x_{7} x_{3} + x_{7} x_{1} x_{4} + x_{7} x_{2} x_{5} + x_{7} x_{6} x_{3} - 3 & = & 0\\
    x_{7} x_{4} + x_{7} x_{1} x_{5} + x_{7} x_{2} x_{6} - 4 & = & 0\\
    x_{7} x_{5} + x_{7} x_{1} x_{6} - 5 & = & 0\\
    x_{6} x_{7} - 6 & = & 0\\
    x_{1} + x_{2} + x_{3} + x_{4} + x_{5} + x_{6} + 1 & = & 0\\
    \end{eqnarray*}

\item eco8: Number of variables:8,Number of equations:8,Max degree:3
    \begin{eqnarray*}
    x_{1} x_{8} + x_{8} x_{1} x_{2} + x_{8} x_{2} x_{3} + x_{8} x_{3} x_{4} + x_{8} x_{4} x_{5} & &\\
    + x_{8} x_{5} x_{6} + x_{8} x_{6} x_{7} - 1 & = & 0\\
    x_{2} x_{8} + x_{8} x_{1} x_{3} + x_{8} x_{2} x_{4} + x_{8} x_{3} x_{5} + x_{8} x_{6} x_{4} + x_{8} x_{7} x_{5} - 2 & = & 0\\
    x_{8} x_{3} + x_{8} x_{1} x_{4} + x_{8} x_{2} x_{5} + x_{8} x_{6} x_{3} + x_{8} x_{7} x_{4} - 3 & = & 0\\
    x_{8} x_{4} + x_{8} x_{1} x_{5} + x_{8} x_{2} x_{6} + x_{8} x_{7} x_{3} - 4 & = & 0\\
    x_{8} x_{5} + x_{8} x_{1} x_{6} + x_{8} x_{7} x_{2} - 5 & = & 0\\
    x_{8} x_{6} + x_{8} x_{7} x_{1} - 6 & = & 0\\
    x_{7} x_{8} - 7 & = & 0\\
    x_{1} + x_{2} + x_{3} + x_{4} + x_{5} + x_{6} + x_{7} + 1 & = & 0\\
    \end{eqnarray*}

\item geneig: Number of variables:6,Number of equations:6,Max degree:3
    \begin{eqnarray*}
     - 10 x_{1} x_{6}^{2} + 2 x_{2} x_{6}^{2} - x_{3} x_{6}^{2} + x_{4} x_{6}^{2} + 3 x_{5} x_{6}^{2} + x_{1} x_{6} + 2 x_{2} x_{6} & &\\
      + x_{3} x_{6} + 2 x_{4} x_{6} + x_{5} x_{6} + 10 x_{1} + 2 x_{2} - x_{3} + 2x_{4} - 2x_{5} & = & 0\\
    2 x_{1} x_{6}^{2} - 11 x_{2} x_{6}^{2} + 2 x_{3} x_{6}^{2} - 2 x_{4} x_{6}^{2} + x_{5} x_{6}^{2} + 2 x_{1} x_{6} + x_{2} x_{6} & &\\
     + 2 x_{3} x_{6} + x_{4} x_{6} + 3x_{5}x_{6} + 2 x_{1} + 9 x_{2} + 3 x_{3} - x_{4} - 2 x_{5} & = & 0\\
     - x_{1} x_{6}^{2} + 2 x_{2} x_{6}^{2} - 12 x_{3} x_{6}^{2} - x_{4} x_{6}^{2} + x_{5} x_{6}^{2} + x_{1} x_{6} + 2 x_{2} x_{6} & &\\
      - 2 x_{4} x_{6} - 2 x_{5} x_{6} - x_{1} + 3 x_{2} + 10 x_{3} + 2 x_{4} - x_{5} & = & 0\\
    x_{1} x_{6}^{2} - 2 x_{2} x_{6}^{2} - x_{3} x_{6}^{2} - 10 x_{4} x_{6}^{2} + 2 x_{5} x_{6}^{2} + 2 x_{1} x_{6} + x_{2} x_{6} & &\\
     - 2 x_{3} x_{6} + 2 x_{4} x_{6} + 3 x_{5} x_{6} + 2x_{1} - x_{2} + 2x_{3} + 12x_{4} + x_{5} & = & 0\\
    3 x_{1} x_{6}^{2} + x_{2} x_{6}^{2} + x_{3} x_{6}^{2} + 2 x_{4} x_{6}^{2} - 11 x_{5} x_{6}^{2} + x_{1} x_{6} + 3 x_{2} x_{6} & &\\
    - 2 x_{3} x_{6} + 3x_{4}x_{6} + 3 x_{5} x_{6} - 2 x_{1} - 2 x_{2} - x_{3} + x_{4} + 10 x_{5} & = & 0\\
    x_{1} + x_{2} + x_{3} + x_{4} + x_{5} - 1 & = & 0\\
    \end{eqnarray*}

\item kinema: Number of variables:9,Number of equations:9,Max degree:2
    \begin{eqnarray*}
    z_{1}^{2} + z_{2}^{2} + z_{3}^{2} - 12 z_{1} - 68 & = & 0\\
    z_{4}^{2} + z_{5}^{2} + z_{6}^{2} - 12 z_{5} - 68 & = & 0\\
    z_{7}^{2} + z_{8}^{2} + z_{9}^{2} - 24 z_{8} - 12 z_{9} + 100 & = & 0\\
    z_{1} z_{4} + z_{2} z_{5} + z_{3} z_{6} - 6 z_{1} - 6 z_{5} - 52 & = & 0\\
    z_{1} z_{7} + z_{2} z_{8} + z_{3} z_{9} - 6 z_{1} - 12 z_{8} - 6 z_{9} + 64 & = & 0\\
    z_{4} z_{7} + z_{5} z_{8} + z_{6} z_{9} - 6 z_{5} - 12 z_{8} - 6 z_{9} + 32 & = & 0\\
    2 z_{2} + 2 z_{3} - z_{4} - z_{5} - 2 z_{6} - z_{7} - z_{9} + 18 & = & 0\\
    z_{1} + z_{2} + 2 z_{3} + 2 z_{4} + 2 z_{6} - 2 z_{7} + z_{8} - z_{9} - 38 & = & 0\\
    z_{1} + z_{3} - 2 z_{4} + z_{5} - z_{6} + 2 z_{7} - 2 z_{8} + 8 & = & 0\\
    \end{eqnarray*}

\item reimer4: Number of variables:4,Number of equations:4,Max degree:5
    \begin{eqnarray*}
    2 x_{1}^{2} - 2 x_{2}^{2} + 2 x_{3}^{2} - 2 x_{4}^{2} - 1 & = & 0\\
    2 x_{1}^{3} - 2 x_{2}^{3} + 2 x_{3}^{3} - 2 x_{4}^{3} - 1 & = & 0\\
    2 x_{1}^{4} - 2 x_{2}^{4} + 2 x_{3}^{4} - 2 x_{4}^{4} - 1 & = & 0\\
    2 x_{1}^{5} - 2 x_{2}^{5} + 2 x_{3}^{5} - 2 x_{4}^{5} - 1 & = & 0\\
    \end{eqnarray*}

\item reimer5: Number of variables:5,Number of equations:5,Max degree:6
    \begin{eqnarray*}
    2 x_{1}^{2} - 2 x_{2}^{2} + 2 x_{3}^{2} - 2 x_{4}^{2} + 2 x_{5}^{2} - 1 & = & 0\\
    2 x_{1}^{3} - 2 x_{2}^{3} + 2 x_{3}^{3} - 2 x_{4}^{3} + 2 x_{5}^{3} - 1 & = & 0\\
    2 x_{1}^{4} - 2 x_{2}^{4} + 2 x_{3}^{4} - 2 x_{4}^{4} + 2 x_{5}^{4} - 1 & = & 0\\
    2 x_{1}^{5} - 2 x_{2}^{5} + 2 x_{3}^{5} - 2 x_{4}^{5} + 2 x_{5}^{5} - 1 & = & 0\\
    2 x_{1}^{6} - 2 x_{2}^{6} + 2 x_{3}^{6} - 2 x_{4}^{6} + 2 x_{5}^{6} - 1 & = & 0\\
    \end{eqnarray*}

\item virasoro: Number of variables:8,Number of equations:8,Max degree:2
    \begin{eqnarray*}
    2 x_{1} x_{4} - 2 x_{1} x_{7} + 2 x_{2} x_{4} - 2 x_{2} x_{6} + 2 x_{3} x_{4} - 2 x_{3} x_{5} + 8 x_{4}^{2} & &\\
    + 8 x_{4} x_{5} + 2 x_{4} x_{6} + 2 x_{4} x_{7} + 6 x_{4} x_{8} - 6 x_{5} x_{8} - x_{4} & = & 0\\
    2 x_{1} x_{5} - 2 x_{1} x_{6} + 2 x_{2} x_{5} - 2 x_{2} x_{7} - 2 x_{3} x_{4} + 2 x_{3} x_{5} + 8 x_{4} x_{5} & &\\
    - 6 x_{4} x_{8} + 8 x_{5}^{2} + 2 x_{5} x_{6} + 2 x_{5} x_{7} + 6 x_{5} x_{8} - x_{5} & = & 0\\
     - 2 x_{1} x_{5} + 2 x_{1} x_{6} - 2 x_{2} x_{4} + 2 x_{2} x_{6} + 2 x_{3} x_{6} - 2 x_{3} x_{7} + 2 x_{4} x_{6} & &\\
     + 2 x_{5} x_{6} + 8 x_{6}^{2} + 8 x_{6} x_{7} + 6 x_{6} x_{8} - 6 x_{7} x_{8} - x_{6} & = & 0\\
     - 2 x_{1} x_{4} + 2 x_{1} x_{7} - 2 x_{2} x_{5} + 2 x_{2} x_{7} - 2 x_{3} x_{6} + 2 x_{3} x_{7} + 2 x_{4} x_{7} & &\\
     + 2 x_{5} x_{7} + 8 x_{6} x_{7} - 6 x_{6} x_{8} + 8 x_{7}^{2} + 6 x_{7} x_{8} - x_{7} & = & 0\\
    8 x_{1}^{2} + 8 x_{1} x_{2} + 8 x_{1} x_{3} + 2 x_{1} x_{4} + 2 x_{1} x_{5} + 2 x_{1} x_{6} & &\\
    + 2 x_{1} x_{7} - 8 x_{2} x_{3} - 2 x_{4} x_{7} - 2 x_{5} x_{6} - x_{1} & = & 0\\
    8 x_{1} x_{2} - 8 x_{1} x_{3} + 8 x_{2}^{2} + 8 x_{2} x_{3} + 2 x_{2} x_{4} + 2 x_{2} x_{5} & &\\
    + 2 x_{2} x_{6} + 2 x_{2} x_{7} - 2 x_{4} x_{6} - 2 x_{5} x_{7} - x_{2} & = & 0\\
     - 8 x_{1} x_{2} + 8 x_{1} x_{3} + 8 x_{2} x_{3} + 8 x_{3}^{2} + 2 x_{3} x_{4} + 2 x_{3} x_{5} & &\\
     + 2 x_{3} x_{6} + 2 x_{3} x_{7} - 2 x_{4} x_{5} - 2 x_{6} x_{7} - x_{3} & = & 0\\
     - 6 x_{4} x_{5} + 6 x_{4} x_{8} + 6 x_{5} x_{8} - 6 x_{6} x_{7} + 6 x_{6} x_{8} + 6 x_{7} x_{8} + 8 x_{8}^{2} - x_{8} & = & 0\\
    \end{eqnarray*}

\end{enumerate}

\flushbottom

\end{document}